\documentclass[11pt]{article}
\usepackage[utf8]{inputenc}
\usepackage{hyperref}
\usepackage{fullpage}
\usepackage{amssymb}
\usepackage{mathtools}
\usepackage{amsthm}
\usepackage{eucal}
\usepackage{dsfont}
\usepackage[T1]{fontenc}
\usepackage{lmodern}
\usepackage[stretch=10,shrink=10]{microtype}
\usepackage[capitalise]{cleveref}
\usepackage{color,soul}
\usepackage{tikz-cd}

\allowdisplaybreaks[1]

\theoremstyle{plain}
\newtheorem{theorem}{Theorem}[section]
\newtheorem{lemma}[theorem]{Lemma}

\newtheorem{cor}[theorem]{Corollary}

\theoremstyle{definition}
\newtheorem{definition}[theorem]{Definition}

\newtheorem{fact}[theorem]{Fact}

\newcommand {\eps} {\varepsilon}

\newcommand {\defeq} {\ensuremath{ \stackrel{\mathrm{def}}{=} }}

\newcommand {\email} [1] {\href{mailto:#1}{\texttt{#1}}}

\newcommand {\ket} [1] {\ensuremath{ \left| #1 \right\rangle }}
\newcommand {\ketbratwo} [2] {\ensuremath{ \left| #1 \middle\rangle \middle\langle #2 \right| }}
\newcommand {\ketbra} [1] {\ketbratwo{#1}{#1}}

\newcommand {\Tr} {\ensuremath{ \mathrm{Tr} }}

\newcommand {\suppress}[1]{}

\def\C{\CMcal{C}}

\def\P{\mathrm{P}}
\def\Q{\CMcal{Q}}
\def\S{\CMcal{S}}
\def\T{\CMcal{T}}

\def\H{\CMcal{H}}

\def\L{\mathcal{L}}
\def\N{\mathcal{N}}

\newcommand {\mytitle} {Concentration bounds for quantum states with finite correlation length on quantum spin lattice systems}
\newcommand {\Anurag}  {Anurag Anshu}

\newcommand {\CQT} {Centre for Quantum Technologies}
\newcommand {\NUS} {National University of Singapore}

\newcommand {\authorblock} [3] {
	\begin{minipage}[t]{0.3\linewidth}
		\centering
		{#1}\\[0.8ex]
		{\footnotesize {#2}\\[-0.7ex]
		\email{#3}}
	\end{minipage}\vspace{1ex}
}

\hypersetup{
  colorlinks={true},
	citecolor={blue},
	pdfstartview={FitH},
	pdfdisplaydoctitle={true},
	breaklinks={true},
	bookmarksopen={true},
	bookmarksnumbered={false},
	pdftitle={\mytitle},
	pdfauthor={\Anurag}
}

\begin{document}

\title{\textbf{\mytitle}\\[2ex]}

\author{
    \authorblock{\Anurag}{\CQT, \NUS}{a0109169@u.nus.edu}
}

\maketitle

\begin{abstract}
We consider the problem of determining the energy distribution of quantum states that satisfy exponential decay of correlation and product states, with respect to a quantum local hamiltonian on a spin lattice. For a quantum state on a $D$-dimensional lattice that has correlation length $\sigma$ and has average energy $e$ with respect to a given local hamiltonian (with $n$ local terms, each of which has norm at most $1$), we show that the overlap of this state with eigenspace of energy $f$ is at most $exp(-((e-f)^2\sigma)^{\frac{1}{D+1}}/n^{\frac{1}{D+1}}D\sigma)$. This bound holds whenever $|e-f|>2^{D}\sqrt{n\sigma}$. Thus, on a one dimensional lattice, the tail of the energy distribution decays exponentially with the energy. 

For product states, we improve above result to obtain a Gaussian decay in energy, even for quantum spin systems without an underlying lattice structure. Given a product state on a collection of spins which has average energy $e$ with respect to a local hamiltonian (with $n$ local terms and each local term overlapping with at most $m$ other local terms), we show that the overlap of this state with eigenspace of energy $f$ is at most $exp(-(e-f)^2/nm^2)$. This bound holds whenever $|e-f|>m\sqrt{n}$. 
\end{abstract}
\section{Introduction}
\label{sec:intro}

A question of primary interest for local hamiltonian spin systems is to determine the energy distribution of natural class of states with respect to a given local hamiltonian. The knowledge of energy distribution reveals a lot of information about the nature of the state itself. As we shall discuss below, a \textit{gaussian distribution} of energy can be associated to a product state. On the other hand, the well known entangled state $\frac{1}{\sqrt{2}}\ket{0}^{\otimes n} + \frac{1}{\sqrt{2}}\ket{1}^{\otimes n}$ (also termed as the `cat state') has energy distribution peaking at opposite ends of the spectrum of the hamiltonian: $\sum_{i=1}^n \ketbra{1}_i$. Moreover, the knowledge of energy distribution plays an important role in the study of thermalization of quantum systems.

The aforementioned question has been well studied in classical setting, important examples of which are the Chernoff bound \cite{chernoff} and the Central limit theorem (which applies to asymptotic regime). Chernoff bound can be informally states as follows. Let $X_1,X_2\ldots X_n$ be independent and identically distributed random variables taking values in $[0,1]$ and each having average value $A$. Then $\text{Pr}(|X_1+X_2\ldots +X_n - nA| > \eps) \leq e^{-c\eps^2/n}$, where $c$ is a constant that depends on $A$. 

One interpretation of this bound (which was the original motivation in \cite{chernoff}) is that it provides a recipe for distinguishing between two probability distributions $P\defeq \sum_x p(x)\ketbra{x}$ and $Q\defeq \sum_x q(x)\ketbra{x}$ with expectation values $A$ and $B$ respectively. Given $n$ independent samples $x_1,x_2\ldots x_n$ from either of these distributions, the sum $\sum_i x_i$ is highly likely to be concentrated around $nA$ if the underlying distribution is $P$ and around $nB$ if the underlying distribution is $Q$. A more precise formulation of this idea requires characterizing the trace distance between $P^{\otimes n}$ and $Q^{\otimes n}$ as $n$ becomes large, and it has been generalized to the quantum setting in \cite{audernaut}.  

Another interpretation of the Chernoff bound, which is the focus of present work, lies in the setting of `classical' local Hamiltonian systems. Consider a product state $\rho^{\otimes n}$ on $n$ sites, where $\rho\defeq \sum_x p(x)\ketbra{x}$. Let $H$ be a $1$-local Hamiltonian $H\defeq \sum_i h_i$, such that $h_i = \sum_x x\ketbra{x}$ acts non-trivially only on the site $i$ and is same for each site. If $A \defeq \Tr(\rho h_i)$ is the expectation value of $\rho$ with respect to $h_i$, then $nA$ is \textit{average energy} of $\rho^{\otimes n}$ with respect to the hamiltonian $H$. Let $\Pi_{\geq nA+\eps}$ be the projector onto eigenstates of $H$ with energy at least $nA+\eps$. Then the Chernoff bound implies that $\Tr(\rho^{\otimes n}\Pi_{\geq nA+\eps}) \leq e^{-c\eps^2/n}$. Thus, the energy distribution of $\rho^{\otimes n}$ is highly concentrated around the average energy $nA$.   

\suppress{
\bigskip
\begin{figure}[ht]
\centering
\begin{tikzpicture}[xscale=1,yscale=1.2]

\draw[thick, <->] (0,4) -- (0,0) -- (7.5,0);
\draw[very thick, blue] (0.5,0.5) to [out=5, in=200] (2.6,3);
\draw[very thick, blue] (2.6,3) to [out=10, in=170] (2.8,3);
\draw[very thick, blue] (2.8,3) to [out=340, in=175] (4.9,0.5);

\draw[very thick, red] (3.0,0.5) to [out=5, in=200] (5.1,3);
\draw[very thick, red] (5.1,3) to [out=10, in=170] (5.3,3);
\draw[very thick, red] (5.3,3) to [out=340, in=175] (7.4,0.5);

\draw [thick] (2.75,0.1) -- (2.75,-0.1);
\node at (2.8,-0.3) {$nA$};

\draw [thick] (5.35,0.1) -- (5.35,-0.1);
\node at (5.4,-0.3) {$nB$};

\draw[thick, <->] (8.5,4) -- (8.5,0) -- (13.5,0);
\draw[very thick, purple] (9,0.5) to [out=5, in=200] (11.1,3);
\draw[very thick, purple] (11.1,3) to [out=10, in=170] (11.3,3);
\draw[very thick, purple] (11.3,3) to [out=340, in=175] (13.4,0.5);

\draw [thick] (11.25,0.1) -- (11.25,-0.1);
\node at (11.3,-0.3) {$nA$};

\node at (13.3,-0.3) {$e$};
\node at (7.9,4) {$\Tr(\rho\Pi_e)$};

\end{tikzpicture}
\caption{Left figure: Plot of the two distributions $X^{\otimes n}$ (blue) and $Y^{\otimes n}$ (red) against  with expectation values $nA$ and $nB$ respectively. The probability associated to their intersection is exponentially small in $n$, }
 \label{fig:chernoff}
\end{figure}
\bigskip
}
The energy distribution of a product state for quantum lattice system with infinitely many sites was considered in \cite{goderis} (for translationally invariant systems) and in \cite{hartmann} (for non-translationally invariant systems). These results can be regarded as a generalization of the Central limit theorem to quantum systems. A non-asymptotic version of Central limit theorem is the Berry-Esseen theorem (\cite{berry},\cite{esseen}), which provides an upper bound on \textit{trace distance} between energy distribution of product state and the \textit{normal distribution} as a function of lattice size. This upper bound goes to zero as lattice size approaches infinity, thus recovering the Central limit theorem.  For quantum states with finite correlation length (which includes product states) on finite sized lattice, a quantum version of Berry-Esseen theorem  was recently shown to hold in \cite{brandao1},\cite{brandao2}. 

These results give a strong indication that states satisfying exponential decay of correlation behave similar to product states, even when their energy distributions are measured with respect to the eigenspectrum of a non-commuting (but local) hamiltonian. The work \cite{Keating2015} goes even further to show that non-commuting local hamiltonians themselves have energy spectrum that resemble that of a $1$-local hamiltonian (although, quite curiously, the same work shows that almost all eigenvectors of non-commuting local hamiltonians are highly entangled, in contrast with the eigenvectors of $1$-local hamiltonians). 

Above mentioned results have added to the growing body of research on general properties of local hamiltonian systems, such as the Lieb-Robinson bound \cite{liebrobinson}, exponential decay of correlation \cite{hastings04}, the area laws \cite{Hastings2007, ALV2012, AKLV2013} and local reversibility \cite{aradkuwahara}, to name a few. They have also found several applications in the problem of thermalization of many body systems. To start with, one of the first steps towards the problem of \textit{locality of temperature} \footnote{which is roughly the problem of assigning a temperature to reduced density matrix of  thermal state of a local hamiltonian, detailed discussion can be found in \cite{kliesch13}} was taken in \cite{Hartmann04}. Crucially using the Central limit theorem obtained in \cite{hartmann}, the authors characterized a set of conditions under which a given thermal state of a quantum local hamiltonian on a lattice would be close to a tensor product of thermal states on local subsystems on the lattice. 

The work \cite{Cramer12} considered the problem of thermalization under random hamiltonians, where the hamiltonian was generated via a random unitary on a fixed local hamiltonian $H$. One of the main technical problems in this work was the study of the characteristic function $\Tr(e^{iHt}\frac{I}{d})$, where $\frac{I}{d}$ is the maximally mixed state (which is also a product state on the lattice). The techniques were inspired from the proof of Central limit theorem in the works \cite{goderis},\cite{hartmann}, where the characteristic function $\Tr(e^{iHt}\rho)$ of a product state $\rho$ had been investigated in detail. 

The quantum version of Berry-Esseen theorem \cite{brandao2} was used to show in \cite{brandao1} that Gibbs state of a local hamiltonian $H$ at sufficiently high temperature (high enough to ensure a clustering of correlation) is indistinguishable, over sufficiently large regions of lattice (that scale sub-linearly with lattice size), from the microcanonical ensemble of eigenstates of $H$ which have eigenvalues close to the average energy of the Gibbs state. This result bears close resemblance to the Eigenstate Thermalization Hypothesis \cite{Srednicki94, Deutsch91}, which is a stronger conjecture stating that every eigenstate of $H$ with eigenvalue close to the average energy of Gibbs state of $H$ is locally indistinguishable from this Gibbs state.   

In present work, we provide further details on energy distribution of states satisfying exponential decay of correlation and product states. Our main results can be seen as an analogue of the Chernoff bound for quantum lattice systems.

Our first result concerns states that satisfy exponential decay of correlation on a $D$-dimensional lattice. Well known examples of such states include the ground states of gapped local hamiltonians \cite{hastings04} and Gibbs state above a finite temperature \cite{kliesch13}. In fact, it has been shown in \cite{Friesdorf15} that for local hamiltonians exhibiting many body localization and having non-degenerate energy spectrum, all eigenvectors satisfy exponential decay of correlation. Thus our result provides information about structure of eigenvectors of such hamiltonians and may have applications in the phenomena of many body localization. 

Fix a $D$-dimensional lattice with spins of arbitrary local dimension sitting on each lattice site. Consider local hamiltonian $H$ on the lattice with $n$ local interaction terms, such that each local term has operator norm at most $1$ and its support is a hyper-cube of side length $2k$, hence containing $(2k)^D$ lattice sites  (see Section \ref{sec:preliminaries} and Figure \ref{fig:quantumlattice} for a detailed description of $H$). 
\begin{theorem}[Informal]
\label{theo:expodecay}
 Let $\rho$ be a quantum state with correlation length $\sigma$ and $\langle H\rangle_{\rho} \defeq \Tr(\rho H)$ be the average energy of $\rho$. Let $\Pi_{\geq f}$ ($\Pi_{\leq f}$) be the projection onto subspace which is union of eigenspaces of $H$ with eigenvalues $\geq f$ ($\leq f$).

For $a\geq \sqrt{\frac{2^{\mathcal{O}(D\log k)}}{n\sigma}}$ it holds that, 
$$\Tr(\rho\Pi_{\geq \langle H\rangle_{\rho} + na}) \leq \mathcal{O}(\sigma) e^{\frac{2Dk}{\sigma}}\cdot e^{-\frac{(na^2\sigma)^{\frac{1}{D+1}}}{\mathcal{O}(1)D\sigma}} \text{ and } \Tr(\rho\Pi_{\leq \langle H\rangle_{\rho} - na}) \leq \mathcal{O}(\sigma) e^{\frac{2Dk}{\sigma}}\cdot e^{-\frac{(na^2\sigma)^{\frac{1}{D+1}}}{\mathcal{O}(1)D\sigma}}.$$
\end{theorem} 

Formal statement of the theorem is given in Theorem \ref{formaltheo:expodecay}. Thus in one dimensional spin chain (with $D=1$), our upper bound decays exponentially with energy, rather than as a gaussian. The bound becomes weaker with higher dimensions and is depicted in Figure \ref{fig:distribution}. 

Our second result concerns product states over a collection of spins and does not require any underlying lattice arrangement of these spins. It does impose, however, a locality constraint on the hamiltonian that acts on these spins. Consider a hamiltonian $H$ which is a sum of $n$ terms, each term being $k$-local (that is, it acts non-trivially on at most $k$ spins) and having operator norm at most $1$. Let $m$ be the maximum number of neighbours of any local term, where two local terms are neighbours if there is a spin on which both act non-trivially (See Section \ref{sec:momentbound} and Figure \ref{fig:spincollection} for detailed description of $H$). We show the following.

\begin{theorem}[Informal]
\label{theo:qchernoff}
Consider a product state $\rho$ with average energy $\langle H\rangle_{\rho} \defeq \Tr(\rho H)$. Fix a real number $a\geq \sqrt{\frac{\mathcal{O}(m^2)}{n}}$. Let $\Pi_{\geq f}$ ($\Pi_{\leq f}$) be the projection onto subspace which is union of eigenspaces of $H$ with eigenvalues $\geq f$ ($\leq f$). It holds that $$\Tr(\rho\Pi_{\geq \langle H\rangle + na}) \leq e^{-\frac{na^2}{\mathcal{O}(m^2)}}$$ and $$\Tr(\rho\Pi_{\leq \langle H\rangle - na}) \leq e^{-\frac{na^2}{\mathcal{O}(m^2)}}.$$
\end{theorem} 

Formal statement of the theorem is given in Theorem \ref{formaltheo:qchernoff}. The energy distribution is depicted in Figure \ref{fig:distribution}. The bound is not only independent of any underlying lattice structure, but is also independent of the locality $k$. This is not surprising, since the quantity $n$ that appears in the bound is the number of local terms in $H$, rather than number of spins on which $H$ acts. Following corollary is a restatement of above bound, in terms of number of spins (which we call $N$) and the maximum number of local terms that act on any given spin (which we call $g$). In the following, we also assume that each local term is exactly $k$-local. 

\begin{cor}
\label{cor:qchernoff}
Consider a product state $\rho$ with average energy $\langle H\rangle_{\rho} \defeq \Tr(\rho H)$. Fix a real number $\eps\geq \sqrt{\mathcal{O}(g^3kN)}$. Let $\Pi_{\geq f}$ ($\Pi_{\leq f}$) be the projection onto subspace which is union of eigenspaces of $H$ with eigenvalues $\geq f$ ($\leq f$). It holds that $$\Tr(\rho\Pi_{\geq \langle H\rangle + \eps}) \leq e^{-\frac{\eps^2}{\mathcal{O}(g^3kN)}}$$ and $$\Tr(\rho\Pi_{\leq \langle H\rangle - na}) \leq e^{-\frac{\eps^2}{\mathcal{O}(g^3kN)}}.$$
\end{cor} 

Formal statement of the corollary appears as Corollary \ref{formalcor:qchernoff}. It shows a gaussian decay for tail of energy distribution of product states in the scenario where $g,k$ are constants \footnote{An interesting class of local hamiltonian system with constant $g,k$ is the family of hamiltonians defined on an expander graph, which has recently been a subject of interest with reference to the Quantum PCP conjecture \cite{Brandao2013,Aharonov2013,eldar2013, Hastings2014,eldar2015,eldarharrow2015}. } independent of $N$ .

\bigskip
\begin{figure}[ht]
\centering
\begin{tikzpicture}[xscale=1.2,yscale=1.4]
\draw[thick, <->] (0,4) -- (0,0) -- (5,0);
\draw[very thick] (0.5,0.5) to [out=5, in=210] (2.4,3);
\draw[very thick] (3,3) to [out=330, in=175] (4.9,0.5);
\path [fill=gray] (0.5,0.5) to [out=5,in=240] (1.4,1.2) -- (1.4,0.5) -- (0.5,0.5);
\path [fill=gray] (4,1.2) to [out=300,in=175] (4.9,0.5) -- (4,0.5) -- (4,1.2);
\draw [thick, <->] (1.45,1) -- (3.95,1);
\draw [thick, <->] (2.4,3.1) -- (3,3.1);
\draw [thick] (2.75,0.1) -- (2.75,-0.1);
\node at (2.8,-0.3) {$\langle H\rangle_{\rho}$};
\node at (2.8,1.3) {$\sqrt{n\log(n)}$};
\node at (2.7,3.4) {$\sqrt{n}$};
\node at (4.8,-0.3) {$e$};
\node at (-0.6,4) {$\Tr(\rho\Pi_e)$};

\draw[thick, <->] (7,4) -- (7,0) -- (13,0);
\draw[very thick] (7.2,0.8) to [out=5, in=250] (9.7,3);
\draw[very thick] (10.3,3) to [out=290, in=175] (12.8,0.8);
\path [fill=gray] (7.2,0.79) to [out=13,in=220] (8.5,1.25) -- (8.5,0.5) -- (7.2,0.5) -- (7.2,0.75);
\path [fill=gray] (11.5,1.25) to [out=320,in=167] (12.8,0.8) -- (12.8,0.5) -- (11.5,0.5) -- (11.5,1.25);
\draw [thick, <->] (8.55,1) -- (11.45,1);
\draw [thick, <->] (9.7,3.1) -- (10.3,3.1);
\draw [thick] (10,0.1) -- (10,-0.1);
\node at (10,-0.3) {$\langle H\rangle_{\rho}$};
\node at (9.9,1.3) {$\sqrt{n(\sigma\log(n))^{D+1}}$};
\node at (10,3.4) {$\sqrt{\frac{n}{\sigma}}$};
\node at (13,-0.3) {$e$};
\node at (6.4,4) {$\Tr(\rho\Pi_e)$};

\end{tikzpicture}
\caption{The tail bounds according to Theorem \ref{theo:expodecay} (right hand side) and Theorem \ref{theo:qchernoff} (left hand side). The $x$-axis is energy $e$ and $y$-axis is the weight $\Tr(\rho\Pi_e)$, where $\Pi_e$ is projector onto eigenspace of $H$ with energy $e$. Shaded region depicts the part of energy distribution with overall weight at most $\frac{1}{n}$. The discontinuous part of the curve is where our results provide no information. We have ignored $\mathcal{O}(1)$ constants in the figure.}
 \label{fig:distribution}
\end{figure}
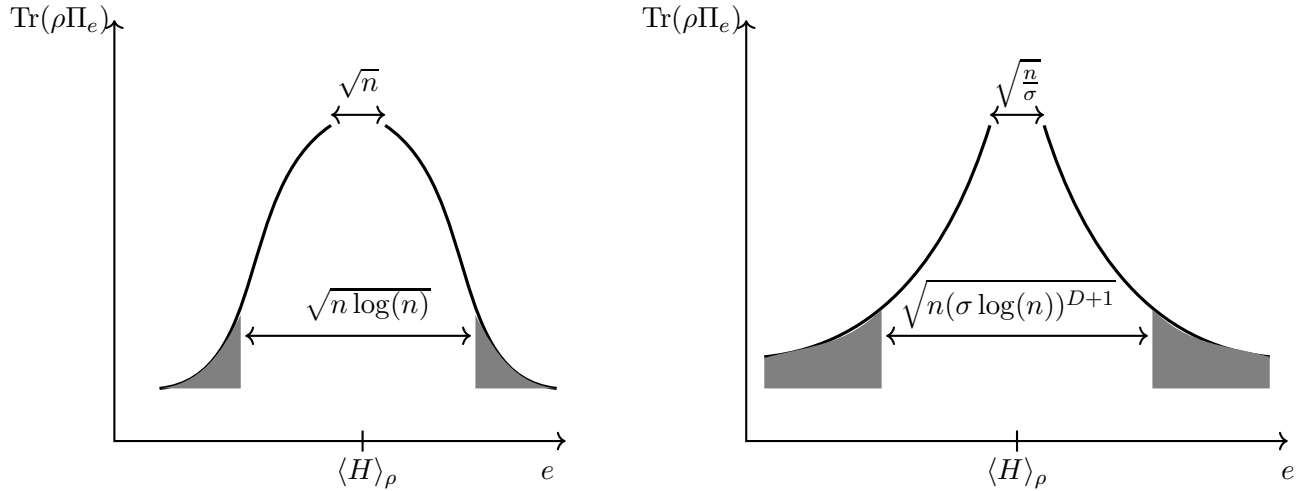
\bigskip

\subsection*{Related recent works}
A recent work \cite{kuwahara} has obtained a similar concentration result for product states (Lemma $4$ therein). The key idea is to split the hamiltonian $H$ as $H=H_1+H_2+\ldots$, where each $H_1,H_2\ldots $ is composed of local terms that are non-overlapping. Then from classical Chernoff bound, the product state exhibits a Gaussian decay in energy distribution for each of the hamiltonians $H_1,H_2\ldots$. Final step (which is also the main argument of the paper) is to combine these tails bounds to obtain a final bound for energy distribution with respect to the original hamiltonian $H$.  Unfortunately, the techniques do not extend to states satisfying exponential decay of correlation. To establish a bound for energy distribution with respect to $H$, one needs the knowledge of bounds for energy distribution with respect to each of the `classical hamiltonians' $H_1,H_2\ldots$. But even for these classical hamiltonians, no bound is known for states that satisfy exponential decay of correlation (apart from Theorem \ref{theo:expodecay}, to the best of author's knowledge). We have provided further comparision of the bound in \cite{kuwahara} and Theorem \ref{theo:qchernoff} in Subsection \ref{subsec:qchernoff}.

A concentration result has been noted in \cite{aradkuwahara} (Section $5$ in this reference) for ground states of gapped local hamiltonians on finite dimensional quantum lattice systems, which also exhibit exponential decay of correlation (\cite{hastings04}). In this work, the probability distribution has been shown to be concentrated about the \textit{median} of the distribution with the weight of the distribution above energy $\eps$ decaying as $e^{-|\eps-f|/\mathcal{O}(1)\sqrt{n\sigma}}$ ($f$ being the median of the distribution, $n$ being the number of local terms in the local hamiltonian and $\sigma$ being the correlation length of the ground state). In comparison, we show a concentration about the \textit{mean} of the distribution for all states satisfying exponential decay of correlation. While our bounds are weaker than those of \cite{aradkuwahara} in higher dimensions, it may be noted that we have considered a larger class of states that might possess weaker properties than the ground states of gapped local hamiltonians. This behaviour appears in the context of area laws as well: ground states of gapped local hamiltonians are known to have very good scaling of area laws with correlation length \cite{AKLV2013}; whereas a recent observation of Hastings \cite{Hastings2015} suggests that states satisfying exponential decay of correlation may have much weaker dependence of area law with correlation length \cite{brandaohorodecki}.   

\subsection*{Our technique and organisation} 

The idea behind our approach is straightforward, to compute the moment generating function $\Tr(H^r\rho)$ of the energy distribution and then use Markov's inequality to upper bound the desired probability. Without loss of generality, we can assume that $H = \sum_w h_w$, where $w$ is a label for local terms and $\langle h_w\rangle_{\rho}\defeq \Tr(\rho h_w)=0$. Our key technical contribution is the combinatorial lemma (Lemma \ref{lem:combinatorial}) which answers the following question: if we expand $H^r$ as a sum of product of local terms, that is $H^r = \sum_{w_1,w_2\ldots w_r}h_{w_1}h_{w_2}\ldots h_{w_r}$, how many terms $h_{w_1}h_{w_2}\ldots h_{w_r}$ make non-negligible (or non-zero) contribution to the moment generating function? We observe the terms that make non-negligible contribution possess a common property: there is no $h_{w_i}$ which is supported `far' from all of $h_{w_1},h_{w_2},\ldots h_{w_{i-1}}, h_{w_{i+1}},\ldots h_{w_r}$. Making the notion of `far' precise, we compute the number of such terms in Lemma \ref{lem:combinatorial} and use it to bound the moment generating function.

The paper is organized as follows. We state basic facts and describe our physical set-up needed for Theorem \ref{theo:expodecay} in section \ref{sec:preliminaries}. We prove our combinatorial lemma in Section \ref{sec:combinatorial}. In Section \ref{sec:expodecay} we prove our bounds for states satisfying exponential decay of correlation. In Section \ref{sec:momentbound}, we introduce the physical set-up required for Theorem \ref{theo:qchernoff} and provide the proof of the theorem. This proof also requires a variant of the combinatorial lemma (Lemma \ref{lem:combinatorial}) which we prove in Appendix \ref{append:combinatorial}. We conclude in Section \ref{sec:conclusion} and address some questions left open by this work. 

\section{Physical set-up and basic facts}
\label{sec:preliminaries}

In this section, we introduce the physical-set up required for Theorem \ref{theo:expodecay}. For simplicity of the presentation, we shall assume that the spins are arranged on a square lattice, with a local interaction term acting between only those spins that are the vertices of a common `unit-hypercube'. We shall introduce the notion of a `dual lattice' below, to formally and concisely represent these local interactions between the spins.  It can be observed that more general local interactions on a square lattice can be put in this form by sufficient coarse-graining of lattice sites. The physical set-up for Theorem \ref{theo:qchernoff} is relatively simple, and shall be introduced directly in Section \ref{sec:momentbound}.

Consider a $D$-dimensional real vector space $\mathbb{R}^D$. For a vector $v\in \mathbb{R}^D$, let $v_i$ represent its $i$-th component. For two vectors $v,v'\in \mathbb{R}^D$, define the `$1$-norm distance' as $$\|v-v'\|\defeq \sum_i |v_i-v'_i|.$$ It satisfies the triangle inequality: given $v,v',v" \in \mathbb{R}^D$, we have $$\|v-v"\|\leq \|v-v'\|+\|v'-v"\|.$$  For brevity, we shall refer to $1$-norm distance simply as \textit{distance}.  

For an integer $L>0$, define a lattice $\L_{D,L}$ as the set of all vectors $v\in \mathbb{R}^D$ that satisfy the following: for all $i\in \{1,2\ldots D\}$, it holds that $v_i$ is an integer and $0\leq v_i\leq L$.  Two vectors $v,v'\in \L_{D,L}$ are \textit{neighbours} if $\|v-v'\|=1$. Henceforth, the vectors belonging to $\L_{D,L}$ shall be referred to as \textit{sites}. 

For each site $v\in \L_{D,L}$, we associate a $d$-dimensional Hilbert space $\H_v^d$  and define the full Hilbert space as $\H = \otimes_{v\in \L_{D,L}}\H_v^d$. Local hamiltonian system is conveniently represented using the notion of \textit{dual lattice}. Let $\bar{\L}_{D,L}$ be the set of vectors $w$ such that for all $i\in \{1,2\ldots D\}$, $0<w_i < L$ and $w_i$ is a half integer (that is, $w_i= k+\frac{1}{2}$, for $k$ an integer). For a fixed $w\in \bar{\L}_{D,L}$ and an integer $k$, let $\S(w,k)$ be the set of all sites $v\in \L_{D,L}$ such that: for all $i\in \{1,2\ldots D\}$, $|v_i - w_i| \leq k+\frac{1}{2}$ 

 A local hamiltonian on $\L_{D,L}$ is defined as $H=\sum_{w\in \bar{\L}_{D,L}} h_w$, where $h_w$ is a `$(2k+2)^D$-local' term that acts non-trivially only on sites in $\S(w,k)$ and acts as identity on rest of the sites. The number of sites in $\S(w,k)$ is at most $(2k+2)^D$, justifying $h_w$ as a `$(2k+2)^D$-local' interaction. Following the physical motivation, we shall refer to vectors in $\bar{\L}_{D,L}$ as \textit{interactions}. Figure \ref{fig:quantumlattice} illustrates the notions introduced above for the case when $D=2$. 

\bigskip
\begin{figure}[ht]
\centering
\begin{tikzpicture}[xscale=1.2,yscale=1.2]

\draw[help lines] (0,0) grid (4,4);
\draw[fill] (0,0) circle [radius=0.05];
\draw[fill] (0,1) circle [radius=0.05];
\draw[fill] (0,2) circle [radius=0.05];
\draw[fill] (0,3) circle [radius=0.05];
\draw[fill] (0,4) circle [radius=0.05];
\draw[fill] (1,0) circle [radius=0.05];
\draw[fill] (1,1) circle [radius=0.05];
\draw[fill] (1,2) circle [radius=0.05];
\draw[fill] (1,3) circle [radius=0.05];
\draw[fill] (1,4) circle [radius=0.05];
\draw[fill] (2,0) circle [radius=0.05];
\draw[fill] (2,1) circle [radius=0.05];
\draw[fill] (2,2) circle [radius=0.05];
\draw[fill] (2,3) circle [radius=0.05];
\draw[fill] (2,4) circle [radius=0.05];
\draw[fill] (3,0) circle [radius=0.05];
\draw[fill] (3,1) circle [radius=0.05];
\draw[fill] (3,2) circle [radius=0.05];
\draw[fill] (3,3) circle [radius=0.05];
\draw[fill] (3,4) circle [radius=0.05];
\draw[fill] (4,0) circle [radius=0.05];
\draw[fill] (4,1) circle [radius=0.05];
\draw[fill] (4,2) circle [radius=0.05];
\draw[fill] (4,3) circle [radius=0.05];
\draw[fill] (4,4) circle [radius=0.05];

\draw[fill=blue,blue] (0.5,0.5) circle [radius=0.05];
\draw[fill=blue,blue] (0.5,1.5) circle [radius=0.05];
\draw[fill=blue,blue] (0.5,2.5) circle [radius=0.05];
\draw[fill=blue,blue] (0.5,3.5) circle [radius=0.05];
\draw[fill=blue,blue] (1.5,0.5) circle [radius=0.05];
\draw[fill=blue,blue] (1.5,1.5) circle [radius=0.05];
\draw[fill=blue,blue] (1.5,2.5) circle [radius=0.05];
\draw[fill=blue,blue] (1.5,3.5) circle [radius=0.05];
\draw[fill=blue,blue] (2.5,0.5) circle [radius=0.05];
\draw[fill=blue,blue] (2.5,1.5) circle [radius=0.05];
\draw[fill=blue,blue] (2.5,2.5) circle [radius=0.05];
\draw[fill=blue,blue] (2.5,3.5) circle [radius=0.05];
\draw[fill=blue,blue] (3.5,0.5) circle [radius=0.05];
\draw[fill=blue,blue] (3.5,1.5) circle [radius=0.05];
\draw[fill=blue,blue] (3.5,2.5) circle [radius=0.05];
\draw[fill=blue,blue] (3.5,3.5) circle [radius=0.05];

\node at (2.7,2.5) {$w$};
\node at (0.7,1.5) {$w'$};
\node at (0.2,2.1) {$v$};
\node at (2.2,2.1) {$v_1$};
\node at (2.8,2.1) {$v_2$};
\node at (2.2,2.85) {$v_3$};
\node at (2.8,2.85) {$v_4$};
\end{tikzpicture}
\caption{A lattice $\L_{2,4}$ in dimension $D=2$ with $L=4$. Black dots represent the lattice sites. Blue dots represent the vectors of the dual lattice $\bar{\L}_{2,4}$, which we call interactions. For the interaction $w$, the set $\S(w,0)$ is $\{v_1,v_2,v_3,v_4\}$. The distance between sites $v$ and $v_3$ is $\|v-v_3\|=3$. The distance between interactions $w$ and $w'$ is $\|w-w'\|=3$.}
 \label{fig:quantumlattice}
\end{figure}
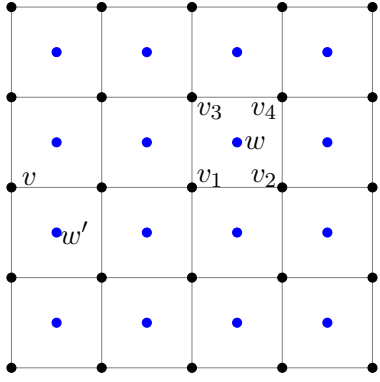
\bigskip

Without loss of generality, we assume that the local terms $h_w$ are positive semi-definite matrices and $\|h_w\|_{\infty}\leq 1$, where $\|.\|_{\infty}$ is the operator norm. Given an operator $A$, \textit{support} of $A$ (called $\text{supp}(A)$) is the set of sites in $\L_{D,L}$  on which $A$ acts non-trivially. We define the \textit{distance} between two operators $A, B$ to be the minimum distance between their respective supports, that is, $\text{min}_{v\in \text{supp}(A), v'\in \text{supp}(B)}\|v-v'\|$.

For a quantum state $\rho \in \H$, the \textit{reduced density matrix} of $\rho$ on a set $T$ of sites is represented as $\rho_T$.  We define the \textit{average energy} of $\rho$ to be $\Tr(\rho H)$, and represent it as $\langle H\rangle_{\rho}$.  For every local term $h_w$, let $\langle h_w\rangle_{\rho} \defeq \Tr(\rho_{\S(w)}h_w)$. Then we have $$\langle H\rangle_{\rho} = \sum_w \langle h_{w}\rangle_{\rho}.$$ 

A state $\rho \in \H$ satisfies $(C,l_0,\sigma)-$ decay of correlation if for any two operators $A,B$ such that distance between $A,B$ is $l\geq l_0$, it holds that $$|\Tr(\rho A\otimes B)-\Tr(\rho A)\Tr(\rho B)|\leq C\|A\|\|B\|e^{-\frac{l}{\sigma}}.$$

Define $\Pi_f$ to be the projector onto eigenspace of $H$ with eigenvalue (energy) equal to $f$. Let $\Pi_{\geq f}$ ($\Pi_{\leq f}$) be the projection onto the subspace which is union of eigenspaces of $H$ with eigenvalues greater (less) than $f$.  The following fact follows from Markov's inequality.  
\begin{fact}
\label{fact:markov}
For every $t,a>0$ and $r$ even, $$\Tr(\rho\Pi_{\geq \langle H\rangle_{\rho} + a}) \leq \frac{\Tr(\rho(H-\langle H\rangle_{\rho})^r)}{(a)^r}.$$
\end{fact}

\begin{proof}
We have $(H-\langle H\rangle_{\rho})^r = \sum_f (f-\langle H\rangle_{\rho})^r\Pi_f$, which gives $\Tr(\rho(H-\langle H\rangle_{\rho})^r) = \sum_f (f-\langle H\rangle_{\rho})^r\Tr(\rho\Pi_f)$. This implies,
$$\Tr(\rho(H-\langle H\rangle_{\rho})^r) \geq  (a)^r\sum_{f>\langle H\rangle_{\rho} + a}\Tr(\rho\Pi_f).$$ 
\end{proof}

\section{A combinatorial lemma}
\label{sec:combinatorial}

In this section, we shall prove a combinatorial lemma, which we shall use in Section \ref{sec:expodecay} to prove Theorem \ref{theo:expodecay}. A slight variant of this lemma shall be proved in Appendix \ref{append:combinatorial} and used in Section \ref{sec:momentbound} to prove Theorem \ref{theo:qchernoff}. We recall the definition of the set $\bar{\L}_{D,L}$ from Section \ref{sec:preliminaries} and let the number of interactions in $\bar{\L}_{D,L}$ be $n$. It is easily seen that $n=(L-1)^D$.

Fix an integer $l$. An ordered set $\{w_1,w_2\ldots w_r\}$ of $r$ interactions in $\bar{\L}_{D,L}$ is said to satisfy a property $\P(l)$ if the following holds: for all $w_i$, there exists a $w_j$ such that $\|w_i-w_j\|\leq l$. Let the number of such ordered sets be $N_D(n,r,l)$.

Rest of the section is devoted to the proof of following lemma.
\begin{lemma}
\label{lem:combinatorial}
It holds that $N_D(n,r,l)\leq (4(4l)^Dnr)^{\frac{r}{2}}$.
\end{lemma}

We start with the following definition that we shall extensively use.
\begin{definition}
A \textit{selection} is an ordered set $\{(b_1,x_1),(b_2,x_2)\ldots (b_r,x_r)\}$, where $b_i\in \{0,1\}$ and $x_i\in \bar{\L}_{D,L}$, that satisfies the following constraints:
\begin{enumerate}
\item If $b_i=0$, then $x_i$ can be any interaction in $\bar{\L}_{D,L}$ and if $b_i=1$, $x_i$ has to satisfy $\|x_i-x_j\|\leq 2\cdot l$ for some $j<i$.
\item Number of $i$ for which $b_i=0$ is at most $\frac{r}{2}$. 
\end{enumerate}
\end{definition}

We show the following lemma from which the proof of Lemma \ref{lem:combinatorial} shall follow immediately.

\begin{lemma}
\label{lem:selection}
Every ordered set $\{w_1,w_2\ldots w_r\}$ that satisfies property $\P(l)$ can be mapped to a \textit{selection} in such a way that for any two distinct sets satisfying $\P(l)$, the corresponding \textit{selections} are distinct. 
\end{lemma}

\begin{proof}
We assign a \textit{selection} to an ordered set $\{v_1,v_2\ldots v_r\}$ satisfying $\P(l)$ using the algorithm below.

\bigskip
\fbox{
\parbox{\textwidth}{
\small

\bigskip 
\textbf{Initialization}
\begin{itemize}
\item Set $i=1$ and $b_i=0, x_i=w_i$.
\item While $(i\leq r)$, do: 
\item $i\rightarrow i+1$. Set $x_i=w_i$. Set $b_i=0$ if there is no $j<i$ such that $\|w_i-w_j\|\leq l$.  Else set $b_i=1$. 
\item End while.
\end{itemize}

\textbf{Pointer creation}
\begin{itemize}
\item Define a relation $R:\{1,2\ldots r\}\rightarrow \{0,1,2\ldots r\}$ as follows. 
\item Set $i=1$. While $(i\leq r)$, do:
\item If $b_i=1$, set $R(i)=0$. If $b_i=0$, find the smallest $j>i$ such that $b_j=1$ and $\|x_j-x_i\|\leq l$ (such a $j$ exists due to property $\P(l)$). Set $R(i)=j$. Set $i \rightarrow i+1$.
\item End while. 
\end{itemize}

\textbf{Update}
\begin{itemize}
\item Let $\mathcal{S}$ be the set of all subsets of $\{1,2,\ldots r\}$ which have cardinality at least $2$. 
\item For each element $S\in \mathcal{S}$, do: 
\item Let $s$ be the cardinality of $S$ and $i_1,i_2\ldots i_s$ be its elements arranged in increasing order. If it holds that $b_{i_1}=b_{i_1}=\ldots b_{i_s}=0$ and $R(i_1)=R(i_2)=\ldots R(i_s) > 0$: set $b_{i_2}=b_{i_3}=\ldots b_{i_s}=1$. 
\item End For.
\end{itemize}
}
}
\bigskip

We show that above algorithm terminates and assigns a \textit{selection} to each ordered set satisfying property $\P(l)$. 
\begin{enumerate}
\item Consider the running of algorithm during the step \textbf{Initialization}. Condition $1$ of a \textit{selection} holds: for every $i$ for which there is a $j<i$ such that $\|x_i-x_j\|\leq l$, we have set $b_i=1$. But we haven't constructed a \textit{selection} yet, since condition $2$ may not be satisfied. 

\item After the step \textbf{Pointer creation}, it may be possible that there exist indices $i_1,i_2\ldots i_s$ (for some $s< r$) such that $b_{i_1}=b_{i_2}=\ldots b_{i_s}=0$, $R(i_1)=R(i_2)=\ldots R(i_s) > 0$ and $i_s> i_{s-1}> \ldots i_1$. In this case, we find using triangle inequality that $\|w_{i_2}-w_{i_1}\|\leq \|w_{i_2}-w_{R(i_2)}\|+\|w_{R(i_2)}-w_{i_1}\|=\|w_{i_2}-w_{R(i_2)}\|+\|w_{R(i_1)}-w_{i_1}\|\leq 2l$. Similarly, $\|w_{i_3}-w_{i_1}\|\leq 2l, |w_{i_4}-w_{i_1}|\leq 2l, \ldots |w_{i_s}-w_{i_1}|\leq 2l$. 

Thus, the step \textbf{Update} sets $b_{i_2}=b_{i_3}=\ldots b_{i_s}=1$, recognizing the fact that each of the points $w_{i_2},w_{i_3}\ldots w_{i_s}$ are at a lattice distance of at most $2l$ from $w_{i_1}$. This ensures that condition $1$ of \textit{selection} is still satisfied.
 
\item After the step \textbf{Update} terminates, condition $2$ of \textit{selection} is now satisfied as well. We now have that for every $i$ with $b_i=0$, there is no other $i'$ such that $R(i)=R(i')$ and $b_i=b_{i'}=0$. Furthermore, $b_{R(i)}=1$. Thus, number of $i$ with $b_i=0$ is at most as large as the number of $j$ with $b_j=1$.
\end{enumerate}

Lemma follows as two distinct ordered sets satisfying $\P(l)$ are not assigned the same \textit{selection}.
\end{proof}
 
Now we prove Lemma \ref{lem:combinatorial}.

\begin{proof}[Proof of Lemma \ref{lem:combinatorial}]
For $n\leq r(4l)^D$, we clearly have $N_D(n,r,l)\leq n^r \leq ((4l)^Dnr)^{\frac{r}{2}} < (4(4l)^Dnr)^{\frac{r}{2}}$.

So we assume $n> r(4l)^D$. We bound the number of \textit{selections}, which gives the desired upper bound on $N_D(n,r,l)$ using Lemma \ref{lem:selection}. 

Consider those \textit{selections} for which number of $i$ such that $b_i=0$ is $u$. For each $i$ with $b_i=0$, number of possible choices of $x_i$ is $n$. For each $i$ with $b_i=1$,  number of possible choices of $x_i$ is at most $(4l)^Dr$ (as there are at most $(4l)^D$ points $x_j\in \L_{D,L}$ that satisfy $\|x_i-x_j\|\leq 2l$ for a given $x_i$ \footnote{This is a very crude upper bound and can be found as follows. The number of non-negative integers $\{a_1,a_2\ldots a_D\}$ such that $\sum_i a_i \leq 2l$ is at most $(2l)^D$. Thus, number of integers $\{a_1,a_2\ldots a_D\}$ such that $\sum_i |a_i| \leq 2l$ is at most $2^D(2l)^D$.}). Hence total number of such \textit{selections} is at most ${r \choose u}n^{u}((4l)^Dr)^{r-u}$. Since $u\leq \frac{r}{2}$, total number of \textit{selections} is at most $$\sum_{u=0}^{\frac{r}{2}}{r \choose u}n^{u}((4l)^Dr)^{r-u} \leq \sum_{u=0}^{\frac{r}{2}}{r \choose u}n^{\frac{r}{2}}((4l)^Dr)^{\frac{r}{2}}<2^r((4l)^Dnr)^{\frac{r}{2}}.$$  

This proves the lemma.
\end{proof}

\section{Energy distribution of states that satisfy an exponential decay of correlation}
\label{sec:expodecay}
Consider a state $\rho$ that satisfies $(C,l_0,\sigma)-$ decay of correlation and the hamiltonian $H=\sum_{w\in \bar{\L}_{D,L}} h_w$, where each term $h_w$ is $(2k+2)^D$-local, that is, it acts non-trivially only on sites in $\S(w,k)$. Let $g_w\defeq h_w - \Tr(\rho h_w)\text{I}$.  We prove following bound on $r$-th moment. 

\begin{lemma}
\label{expodecaymoment}
Given the state $\rho$ that satisfies $(C,l_0,\sigma)-$ decay of correlation, it holds that $$\Tr(\rho(H-\langle H\rangle)^r)\leq (4(4l_0+8Dk)^Dnr)^{\frac{r}{2}}+Ce^{\frac{2Dk}{\sigma}}\sigma\cdot(4(\frac{D\sigma}{2})^Dnr^{D+1})^{\frac{r}{2}}.$$ 
\end{lemma}  
\begin{proof}
Consider,

\begin{eqnarray}
\label{eq:seriesexpansion}
\Tr(\rho (\sum_{w\in \bar{\L}_{D,L}} g_w)^r) &=& \sum_{w_1,w_2\ldots w_r} \Tr(\rho g_{w_1}g_{w_2}\ldots g_{w_r})
\end{eqnarray}

For every ordered set $\{w_1,w_2\ldots w_r\}$ define the quantity $D(w_1,w_2\ldots w_r)\defeq \text{max}_i (\text{min}_{j\neq i}|w_i-w_j|)$. This is the distance of farthest interaction from rest of the interactions in the ordered set.

For an integer $l>0$, define $\T(l)$ as the collection\footnote{to avoid confusion, we call $\T(l)$ a `collection' instead of a `set'} of all sets $\{w_1,w_2\ldots w_r\}$ that satisfy $D(w_1,w_2\ldots w_r)= l$. Now, fix a set $\{w_1,w_2\ldots w_r\}\in \T(l)$. Without loss of generality, suppose that $w_1$ is an interaction at the distance $l$ from rest of the interactions. The distance between operator $g_{w_1}$ and $g_{w_i}$, for any $i\neq 1$, is at least $l- 2Dk$, as the distance from $w_i$ to any site in $\S(w_i,k)$ is at most $Dk$. Then from $(C,\sigma,l_0)-$ decay of correlation and the relation $\Tr(\rho g_{w_1})=0$, it holds that $$\Tr(\rho g_{w_1}g_{w_2}\ldots g_{w_r})\leq \Tr(\rho g_{w_1})\cdot\Tr(\rho g_{w_2}\ldots g_{w_r})+ Ce^{-\frac{l-2Dk}{\sigma}} = Ce^{-\frac{l-2Dk}{\sigma}},$$ as long as $l-2Dk \geq l_0$.

Now, the number of sets in the collection $\T(l)$ is at most $N_D(n,r,l)$ which is upper bounded by $(4(4l)^Dnr)^{\frac{r}{2}}$ (Lemma \ref{lem:combinatorial}). Thus we have

\begin{eqnarray}
\label{eq:expodecaycalc}
\sum_{w_1,w_2\ldots w_r} \Tr(\rho g_{w_1}g_{w_2}\ldots g_{w_r}) &=& \sum_l\sum_{(w_1,w_2\ldots w_r)\in \T(l)} \Tr(\rho g_{w_1}g_{w_2}\ldots g_{w_r}) \nonumber\\&=& \sum_{l\leq l_0+2Dk}\sum_{(w_1,w_2\ldots w_r)\in \T(l)} \Tr(\rho g_{w_1}g_{w_2}\ldots g_{w_r}) \nonumber\\&+& \sum_{l> l_0+2Dk}\sum_{(w_1,w_2\ldots w_r)\in \T(l)} \Tr(\rho g_{w_1}g_{w_2}\ldots g_{w_r}) \nonumber\\ &\leq& N_D(n,r,l_0+2Dk) + \sum_{l>l_0+2Dk} N_D(n,r,l)\cdot Ce^{-\frac{l-2Dk}{\sigma}}\nonumber\\ &\leq& (4(4l_0+8Dk)^Dnr)^{\frac{r}{2}}+Ce^{\frac{2Dk}{\sigma}}\sum_{ l\geq l_0+2Dk} (4(4l)^Dnr)^{\frac{r}{2}}e^{-\frac{l}{\sigma}}\nonumber\\&\leq& (4(4l_0+8Dk)^Dnr)^{\frac{r}{2}}+Ce^{\frac{2Dk}{\sigma}}(4nr)^{\frac{r}{2}}\sum_{l\geq 1} l^{\frac{rD}{2}}e^{-\frac{l}{\sigma}}
\end{eqnarray}
 Now, we evaluate
\begin{eqnarray*}
\sum_{l\geq 1} l^{\frac{rD}{2}}e^{-\frac{l}{\sigma}} &\leq& \int_{0}^{\infty}  l^{\frac{rD}{2}}e^{-\frac{l}{\sigma}}dl = \sigma^{\frac{rD}{2}+1}\int_{0}^{\infty}s^{\frac{rD}{2}}e^{-s}ds \leq \sigma^{\frac{rD}{2}+1}(\frac{rD}{2})^{\frac{rD}{2}}. 
\end{eqnarray*}

Using this in Equation (\ref{eq:expodecaycalc}), we obtain 
 $$\sum_{w_1,w_2\ldots w_r} \Tr(\rho g_{w_1}g_{w_2}\ldots g_{w_r})\leq (4(4l_0+8Dk)^Dnr)^{\frac{r}{2}}+Ce^{\frac{2Dk}{\sigma}}\sigma(4(\frac{D\sigma}{2})^Dnr^{D+1})^{\frac{r}{2}}.$$

\end{proof}

Now we proceed to state Theorem \ref{theo:expodecay} formally and provide its proof. 
\begin{theorem}
\label{formaltheo:expodecay}
Consider a quantum state $\rho$ that satisfies $(C,l_0,\sigma)-$decay of correlation and has average energy $\langle H\rangle_{\rho}$. 

For $\sqrt{\frac{8e(4l_0+8Dk)^{D+1}}{nD\sigma}}\geq a \geq \sqrt{\frac{8e(4l_0+8Dk)^D}{n}}$ (if the range exists) it holds that, 
$$\Tr(\rho\Pi_{\geq \langle H\rangle_{\rho} + na}) \leq (1+C\sigma e^{\frac{2Dk}{\sigma}})e^{-\frac{na^2}{8e(4l_0+8Dk)^D}} \text{ and } \Tr(\rho\Pi_{\leq \langle H\rangle_{\rho} - na}) \leq (1+C\sigma e^{\frac{2Dk}{\sigma}})e^{-\frac{na^2}{8e(4l_0+8Dk)^D}}.$$ For $a\geq \sqrt{\frac{(4l_0+8Dk)^{D+1}}{D\sigma n}}$ it holds that, 
$$\Tr(\rho\Pi_{\geq \langle H\rangle_{\rho} + na}) \leq (1+C\sigma e^{\frac{2Dk}{\sigma}}) e^{-(\frac{na^2}{8e(D\sigma)^D})^{\frac{1}{D+1}}} \text{ and } \Tr(\rho\Pi_{\leq \langle H\rangle_{\rho} - na}) \leq (1+C\sigma e^{\frac{2Dk}{\sigma}}) e^{-(\frac{na^2}{8e(D\sigma)^D})^{\frac{1}{D+1}}}.$$

\end{theorem}
\begin{proof}

Using Fact \ref{fact:markov} and Lemma \ref{expodecaymoment} we have, 
$$\Tr(\rho\Pi_{\geq \langle H\rangle_{\rho} + na}) \leq (\frac{4(4l_0+8Dk)^Dr}{na^2})^{\frac{r}{2}}+Ce^{\frac{2Dk}{\sigma}}\sigma(\frac{4(\frac{D\sigma}{2})^D r^{D+1}}{na^2})^{\frac{r}{2}}$$ 

Consider the following two cases.
\begin{itemize}
\item  $a\geq \sqrt{\frac{(4l_0+8Dk)^{D+1}}{D\sigma n}}$, or equivalently $\frac{(4l_0+8Dk)^{D+1}}{Dna^2\sigma}<1$ 

Then we set 
$r=2\lceil (\frac{na^2}{8e(D\sigma)^D})^{\frac{1}{D+1}} \rceil$, where $\lceil . \rceil$ denotes the ceiling operation (rounding to the nearest larger integer) to obtain

\begin{eqnarray*}
\Tr(\rho\Pi_{\geq \langle H\rangle_{\rho} + na}) &\leq& (\frac{4(4l_0+8Dk)^Dr}{na^2})^{\frac{r}{2}}+Ce^{\frac{2Dk}{\sigma}}\sigma(\frac{4(\frac{D\sigma}{2})^D r^{D+1}}{na^2})^{\frac{r}{2}} \\ &\leq& (\frac{1}{e}(\frac{(4l_0+8Dk)^{D+1}}{Dna^2\sigma})^{\frac{D}{D+1}})^{\frac{r}{2}}+Ce^{\frac{2Dk}{\sigma}}\sigma(\frac{1}{e})^{\frac{r}{2}} \\ &\leq& (1+C\sigma\cdot e^{\frac{2Dk}{\sigma}}) e^{-(\frac{na^2}{8e(D\sigma)^D})^{\frac{1}{D+1}}}
\end{eqnarray*}

The last inequality follows from the assumption: $\frac{(4l_0+8Dk)^{D+1}}{Dna^2\sigma}<1$.

\item $a\leq\sqrt{\frac{8e(4l_0+8Dk)^{D+1}}{nD\sigma}}$ and $a \geq \sqrt{\frac{8e(4l_0+8Dk)^D}{n}}$, or equivalently
$\frac{na^2}{8e(4l_0+8Dk)^D}\geq1$ and \newline $\frac{D\sigma na^2}{8e(4l_0+8Dk)^{(D+1)}}\leq1$.

 We set $r=2\lceil \frac{na^2}{8e(4l_0+8Dk)^D}\rceil$ to obtain
\begin{eqnarray*}
\Tr(\rho\Pi_{\geq \langle H\rangle_{\rho} + na}) &\leq& (\frac{4(4l_0+8Dk)^Dr}{na^2})^{\frac{r}{2}}+Ce^{\frac{2Dk}{\sigma}}\sigma(\frac{4(\frac{D\sigma}{2})^D r^{D+1}}{na^2})^{\frac{r}{2}} \\ &\leq& (\frac{1}{e})^{\frac{r}{2}} + Ce^{\frac{2Dk}{\sigma}}\sigma(\frac{1}{e}(\frac{D\sigma na^2}{8e(4l_0+8Dk)^{(D+1)}})^D)^{\frac{r}{2}}\\ &\leq& (1+C\sigma e^{\frac{2Dk}{\sigma}})e^{-\frac{na^2}{8e(4l_0+8Dk)^D}}
\end{eqnarray*}

Last inequality follows from the assumption: $\frac{D\sigma na^2}{8e(4l_0+8Dk)^{(D+1)}}<1$.
\end{itemize}

For second part of the theorem, consider the hamiltonian $$H'=\sum_{w\in W_{k,m}}\text{I}-h_w.$$ Define $\langle H'\rangle_{\rho} \defeq \Tr(\rho H') = n-\langle H\rangle_{\rho}$. Let $\Pi'_{\geq f}$ be the projector onto subspace with eigenvalues of $H'$ larger than $f$. Same analysis as above for $H'$ in place of $H$, along with the relation $\Pi'_{\geq f} = \Pi_{\leq n-f}$ completes the proof.

\end{proof}

\section{Energy distribution of a product state}
\label{sec:momentbound}

In this section, we introduce the physical set-up for Theorem \ref{theo:qchernoff} and also provide its proof. We shall continue using the notations $H$ and $h$ for the hamiltonian and its local term, as this notation is restricted only to this section. 

Consider a collection $\C$ of spins, such that a $d$-dimensional Hilbert space $\H_s^d$ is associated to each spin $s\in\C$. Let full Hilbert space $\H$ be defined as $\H=\otimes_s \H_s^d$. For an integer $k>0$, let $S_k$ be the set of all subsets of $\C$  of size at most $k$. For an integer $m>0$, let $W_{k,m}$ be a subset of $S_k$ defined as follows (note that $W_{k,m}$ is also a set of subsets of $\C$) : for each $w\in W_{k,m}$ the number of $w'\in W_{k,m}$ such that $|w'\cap w| > 0$ is at most $m$. For each $w\in W_{k,m}$, let $\N(w)$ be the set of all $w'\in W_{k,m}$ such that $|w\cap w'|>0$. Elements of $\N(w)$ shall be referred to as \textit{neighbours} of $w$. The set-up has been depicted in Figure \ref{fig:spincollection}.

\bigskip
\begin{figure}[ht]
\centering
\begin{tikzpicture}[xscale=1.4,yscale=1.4]

\draw[fill=blue,blue] (0.5,1) circle [radius=0.05];
\draw[fill=blue,blue] (0.5,2) circle [radius=0.05];
\draw[fill=blue,blue] (1.3,3) circle [radius=0.05];
\draw[fill=blue,blue] (2.6,3) circle [radius=0.05];
\draw[fill=blue,blue] (3.5,2) circle [radius=0.05];
\draw[fill=blue,blue] (3.5,1) circle [radius=0.05];
\draw[fill=blue,blue] (2.6,0) circle [radius=0.05];
\draw[fill=blue,blue] (1.3,0) circle [radius=0.05];
\draw[fill=blue,blue] (1.4,1) circle [radius=0.05];
\draw[fill=blue,blue] (2.5,1) circle [radius=0.05];
\draw[fill=blue,blue] (1.4,2) circle [radius=0.05];
\draw[fill=blue,blue] (2.5,2) circle [radius=0.05];
\draw[thick, fill=lightgray] (0.5,1) -- (1.3,0) -- (1.4,1) -- (0.5,1);
\draw[thick, fill=lightgray] (1.4,1) -- (0.5,2) -- (1.4,2) -- (1.4,1);
\draw[thick, fill=lightgray] (1.4,2) -- (1.3,3) -- (2.5,2) -- (1.4,2);
\draw[thick, fill=lightgray] (2.5,2) -- (2.6,3) -- (3.5,2) -- (2.5,2);
\draw[thick, fill=lightgray] (3.5,2) -- (3.5,1) -- (2.5,1) -- (3.5,2);
\draw[thick, fill=lightgray] (2.6,0) -- (3.5,1) -- (2.5,1) -- (2.6,0);
\node at (1.1,0.5) {$w_1$};
\node at (1.1,1.6) {$w_2$};
\node at (1.7,2.2) {$w_3$};
\node at (2.9,2.2) {$w_4$};
\node at (3.2,1.3) {$w_5$};
\node at (2.8,0.6) {$w_6$};
\end{tikzpicture}
\caption{Collection of spins (blue dots) with local terms (gray triangles). There is no underlying lattice structure. Each local term is $3$-local, thus $k=3$ and each local term has at most $2$ neighbours, thus $m=2$. The set $W_{3,2}$ in above figure is $\{w_1,w_2,w_3,w_4,w_5,w_6\}$, where each $w_i$ is the set of spins that form vertices of corresponding triangle. Note that for a fixed $k,m$ (here $k=3,m=2$), there can be several choices of the set $W_{k,m}$ and each choice gives a different hamiltonian $H$. Neighbours of $w_4$ are $\N(w_4)=\{w_3,w_5\}$. Each spin in the figure is in the support of at most $2$ local terms. Thus we have $g=2$, where $g$ is defined in Subsection \ref{subsec:qchernoff}.}
 \label{fig:spincollection}
\end{figure}
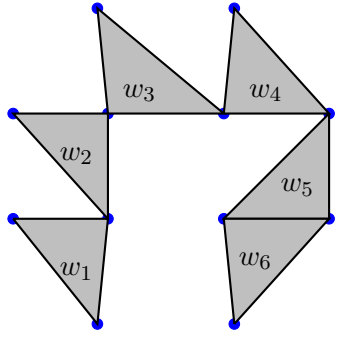
\bigskip

Let the hamiltonian $H$ be defined as: $$H=\sum_{w\in W_{k,m}} h_w,$$ where $h_w$ acts non-trivially only on spins in $w$ and acts trivially on rest of the spins. Further, we assume that $\|h_w\|_{\infty} \leq 1$.

The definition of $W_{k,m}$ thus translates to the assumption that:
\begin{enumerate}
\item Each `local term' $h_w$ acts non-trivially on at most $k$ particle, and hence is $k$-local.
\item For each $h_w$, the number of $h_{w'}$ such that the supports of $h_w$ and $h_{w'}$ intersect, is at most $m$. 
\end{enumerate} 

Let $\rho \in \H$ be a product state, that is, $\rho = \Pi_{s\in \C} \rho_{s}$ and support of each $\rho_s$ is exactly the spin $s$. Let the reduced density matrix of $\rho$ on a subset $T\subseteq \C$ of spins be denoted in the usual way as $\rho_T$. 

We bound the moment function $\Tr(\rho(H-\langle H\rangle_{\rho})^r)$ for an even $r$ to be chosen later and use it to prove Theorem \ref{theo:qchernoff}. Define $g_w\defeq h_{w}-\langle h_{w}\rangle_{\rho} \text{I}$. 

We shall prove the following lemma.
\begin{lemma}
\label{lem:kmomentbound}
Let $n\defeq |W_{k,m}|$ be the number of local terms. Given the product state $\rho = \Pi_{s\in \C} \rho_s$, it holds that
$$\Tr(\rho(\sum_{w\in W_{k,m}}g_w)^r) \leq (4m^2nr)^{\frac{r}{2}}.$$
\end{lemma}

\begin{proof}
Consider,

\begin{eqnarray}
\label{eq:seriesexpansion}
\Tr(\rho(\sum_{w\in W_{k,m}}g_w)^r) &=& \sum_{w_1,w_2\ldots w_r} \Tr(\rho g_{w_1}g_{w_2}\ldots g_{w_r}) 
\end{eqnarray}

Using $\Tr(\rho g_w)= \Tr(\rho_wg_w)=0$, we observe that the term $\Tr(\rho g_{w_1}g_{w_1}\ldots g_{w_r})$ is non-zero only if the ordered set $\{w_1,w_2\ldots w_r\}$ satisfies the following property $\Q$: for every $w_i$, there exists a $w_j$ such that $|w_i \cap w_j|> 0$. In other words, there is a $w_j\in W_{k,m}$ such that $w_i \in \N(w_j)$. Let number of ordered sets $\{w_1,w_2\ldots w_r\}$ that satisfy above property be $N_{k,m}(n,r)$. This gives us

\begin{equation}
\label{eq:boundpower}
\Tr(\rho(\sum_{w\in \bar{\L}_{D,L}}g_w)\leq N_{k,m}(n,r) \text{max}_{w_1,w_2\ldots w_r}|\Tr(\rho g_{w_1}g_{w_2}\ldots g_{w_r})|
\end{equation}

Setting the trivial bound $\text{max}_{w_1,w_2\ldots w_r}|\Tr(\rho g_{w_1}g_{w_2}\ldots g_{w_r})|\leq 1$, and using Lemma \ref{lem:combinatorial1} below, the proof follows.
\end{proof}

\textbf{Remark:} For the case of translationally invariant systems, where $h_w=h$ for all $w$, the bound $\text{max}_{w_1,w_2\ldots w_r}\Tr(\rho g_{w_1}g_{w_2}\ldots g_{w_r})\leq 1$ can be improved to $\text{max}_{w_1,w_2\ldots w_r}\Tr(\rho g_{w_1}g_{w_2}\ldots g_{w_r})|\leq (1-\langle h\rangle_{\rho})^r$. This gives a minor improvement on the statement of Lemma \ref{lem:kmomentbound}.

\bigskip

Now we prove an upper bound on the quantity $N_{k,m}(n,r)$ in the following Lemma. Proof is very similar to Lemma \ref{lem:combinatorial} and is deferred to Appendix \ref{append:combinatorial} for completeness. 
\begin{lemma}
\label{lem:combinatorial1}
It holds that $N_{k,m}(n,r)\leq (4m^2nr)^{\frac{r}{2}}$.
\end{lemma}

We restate Theorem \ref{theo:qchernoff} formally and give its proof below.
\begin{theorem}
\label{formaltheo:qchernoff}
Given the product state $\rho = \Pi_{s\in \C}\rho_s$ with average energy $\langle H\rangle_{\rho}$, consider a real number $a\geq \sqrt{\frac{8em^2}{n}}$. It holds that $$\Tr(\rho\Pi_{\geq \langle H\rangle_{\rho} + na}) \leq e^{-\frac{na^2}{4em^2}}$$ and $$\Tr(\rho\Pi_{\leq \langle H\rangle_{\rho} - na}) \leq e^{-\frac{na^2}{4em^2}}.$$

\end{theorem}

\begin{proof}

Lemma \ref{lem:kmomentbound} gives the following upper bound on $r$-th moment:

\begin{equation}
\label{eq:momentbound}
\Tr(\rho(H-\langle H\rangle_{\rho})^r) \leq (4m^2nr)^{\frac{r}{2}}
\end{equation}

Using Fact \ref{fact:markov}, we have 
\begin{eqnarray*}
\Tr(\rho\Pi_{\geq \langle H\rangle_{\rho} + na}) &\leq& \frac{(4m^2nr)^{\frac{r}{2}}}{(na)^r} = (\frac{4m^2r}{na^2})^{\frac{r}{2}}. 
\end{eqnarray*}

Choosing $r=2\lceil\frac{na^2}{8em^2}\rceil$, we obtain for $a\geq \sqrt{\frac{8m^2e}{n}}$
\begin{eqnarray*}
\Tr(\rho\Pi_{\geq \langle H\rangle_{\rho} + na}) &\leq& (\frac{8m^2\cdot na^2}{8em^2\cdot na^2})^{\frac{r}{2}} \leq e^{-\frac{na^2}{4em^2}}. 
\end{eqnarray*}

For second part of the theorem, consider the hamiltonian $$H'=\sum_{w\in W_{k,m}}\text{I}-h_w.$$ Define $\langle H'\rangle_{\rho} \defeq \Tr(\rho H') = n-\langle H\rangle_{\rho}$. Let $\Pi'_{\geq f}$ be the projector onto subspace with eigenvalues of $H'$ larger than $f$. Same analysis as above for $H'$ in place of $H$ gives
$$\Tr(\rho\Pi'_{\geq \langle H'\rangle_{\rho} + na}) \leq e^{-\frac{na^2}{4em^2}}.$$ This completes the proof since $\Pi'_{\geq f} = \Pi_{\leq n-f}$. 

\end{proof}

\subsection{Restatement of Theorem \ref{formaltheo:qchernoff} in terms of number of spins}
\label{subsec:qchernoff}

We introduce a new parameter that captures the number of local terms that act on any given spin. Define $$g_s \defeq \sum_{w\in W_{k,m}: s\in w}1 , \quad \text{and  } g\defeq \text{max}_s  g_s,$$ where $g_s$ is the maximum number of local terms that act non-trivially on spin $s$.  

Now, we prove Corollary \ref{cor:qchernoff}. Its formal statement is as follows, where we also assume that each local term is exactly $k$-local. 

\begin{cor}
\label{formalcor:qchernoff}
Let the hamiltonian $H$ be such that each term $h_w$ has locality equal to $k$. Let $N \defeq |\C|$ be the number of spins. Given the product state $\rho = \Pi_{s\in \C}\rho_s$ with average energy $\langle H\rangle_{\rho}$, consider a real number $\eps\geq \sqrt{8eg^3kN}$. It holds that $$\Tr(\rho\Pi_{\geq \langle H\rangle_{\rho} + \eps}) \leq e^{-\frac{\eps^2}{4eg^3kN}}$$ and $$\Tr(\rho\Pi_{\leq \langle H\rangle_{\rho} - \eps}) \leq e^{-\frac{\eps^2}{4eg^3kN}}.$$

\end{cor}

\begin{proof}

We set $\eps \defeq na$ as the energy with respect to $H$. Then the bound in Theorem \ref{formaltheo:qchernoff} can be restated as: 
$$\Tr(\rho\Pi_{\geq \langle H\rangle_{\rho}+\eps}) \leq e^{-\frac{\eps^2}{4enm^2}}.$$

Relation between $N$ and $n$ can be computed as follows. To each local term $h_w$, one can associate exactly $k$ spins on which $h_w$ acts non-trivially. On the other hand, to each spin $s$, one can associate at most $g$ local terms that contain $s$ in their support. From the first argument, the number of associations is exactly $k\cdot n$, whereas from the second argument, the number of associations is at most $g\cdot N$. Thus, $g\cdot N\geq k\cdot n$ which implies $n\leq \frac{gN}{k}$.  Also, $m\leq k\cdot g$, since each local term is supported on $k$ spins, and each of these spins are in the support of at most $g$ other local terms. Collectively we obtain $nm^2 \leq Ng^3k$ and our bound takes the form: 

\begin{equation}
\label{eq:qchernoffspinnumber}
\Tr(\rho\Pi_{\geq \langle H\rangle_{\rho}+\eps}) \leq exp(-\frac{\eps^2}{4eg^3kN}).
\end{equation}

This completes the proof.
\end{proof}

Above upper bound may be compared to Theorem $7$ in \cite{kuwahara}. In this reference, the notion of $g'$-extensitivity has been introduced (Definition $2$, \cite{kuwahara}), which is analogous to the locality parameter $g$ defined above. It is defined as follows: A local hamiltonian $H$ is $g'$-extensive if for every spin $s$, we have $\sum_{w\in W_{k,m}: s\in w} \|h_w\|\leq g'$. Using this, the following theorem has been shown in \cite{kuwahara}: 
\begin{theorem}[Informal version of Theorem $7$, \cite{kuwahara}]
\label{theo:kuwahara}
 Given a $g'$-extensive local hamiltonian with locality $k$, it holds that 
$$\Tr(\rho\Pi_{\geq \langle H\rangle_{\rho}+\eps}) \leq \mathcal{O}(1)exp(-\frac{\eps^2}{cN\log(\frac{\eps}{\sqrt{N}})}),$$ where 
$c$ is a $\mathcal{O}(1)$ constant that depends only on $k,g'$.
\end{theorem}   

We observe that Equation (\ref{eq:qchernoffspinnumber}) achieves a marginally better bound whenever the norm of each local term $h_w$ , that is $\|h_w\|$, is a constant independent of $w$. In such a case, $g'$ and $g$ are same up to the norm of local terms. In case the normalizations of each local term are different, it is not clear how $g,g'$ are related to each other. In such a case Equation (\ref{eq:qchernoffspinnumber}) and Theorem \ref{theo:kuwahara} may be viewed as complementary results.

\section{Conclusion}
\label{sec:conclusion}
We have shown upper bounds on tail of energy distribution of states that satisfy exponential decay of correlation and product states, with respect to a local hamiltonian. Main technical tool we use is a combinatorial lemma that gives a non-trivial upper bound on the moments of the energy distribution. The results may have applications in the study of thermalization of many body quantum systems  and also for many body localization, as noted in the Introduction. Main questions that we leave open are connected to tightness of our bounds, as we discuss below.

The bounds presented in Theorem \ref{theo:qchernoff} can only be improved up to constants, since classical Chernoff bound also exhibits a Gaussian decay, which is known to be tight. More interesting situation occurs with the bounds presented in Theorem \ref{theo:expodecay}. In one dimensional spin chain, our bound decays exponential with the energy. For gapped ground states, this is very similar to the behaviour noted in \cite{aradkuwahara} (Section $5$) using completely different techniques. This suggests that gapped ground states (such as the ground state of Transverse field Ising model, which is exactly solvable) are strong candidates for the study of tightness of above results. Our result for higher dimensions appears to be much weaker that those obtained in \cite{aradkuwahara} (Section $5$) for gapped ground states, and we expect further improvement using better combinatorial arguments. 

An another interesting question is with respect to \textit{Matrix product states} (with constant bond dimension) which are defined on one dimensional spin chain. It is well known that under reasonable assumptions (see Section $5.1.1$, \cite{Orus2014}) Matrix product states satisfy exponential decay of correlation. Furthermore, it has already been shown in \cite{Ogata2010} that given a Matrix product state $\rho$, if $n$ is large enough and energy $\eps\approx \mathcal{O}(n)$, it holds that $\Tr(\rho\Pi_{\geq \langle H\rangle_{\rho}+\eps})\leq e^{-\mathcal{O}(n)}$. It is a strong indication that our bound (which applies for all energies $\eps > \mathcal{O}(\sqrt{n}))$ may be considerably improved for this special, but important, class of states.

\section*{Acknowledgement}
I thank Itai Arad for introducing me to the question considered here and helpful discussions. I thank Tomotaka Kuwahara for many insightful discussions and sharing an earlier version of his paper \cite{kuwahara} on the subject considered here. I also thank Fernando Brandao, Rahul Jain and Priyanka Mukhopadhyay for helpful discussions related to proof. 

This work is supported by the Core Grants of the Center for Quantum Technologies (CQT), Singapore. Part of this work was done when author was visiting Institute for Quantum Computing (IQC), University of Waterloo, Waterloo. 

\bibliographystyle{alpha}
\bibliography{references}

\appendix
\section{Proof of Lemma \ref{lem:combinatorial1}}
\label{append:combinatorial}

We repeat most of the proof of Lemma \ref{lem:combinatorial}, making changes wherever necessary. 

We start with the following definition.
\begin{definition}
A \textit{selection} is an ordered set $\{(b_1,x_1),(b_2,x_2)\ldots (b_r,x_r)\}$, where $b_i\in \{0,1\}$ and $x_i\in W_{k,m}$, that satisfies the following constraints:
\begin{enumerate}
\item If $b_i=0$, then $x_i$ can be any element of $W_{k,m}$ and if $b_i=1$, $x_i$ has to satisfy $|\N(x_i)\cap \N(x_j)|>0$ for some $j<i$.
\item Number of $i$ for which $b_i=0$ is at most $\frac{r}{2}$. 
\end{enumerate}
\end{definition}

We show the following Lemma, from which the proof of Lemma \ref{lem:combinatorial} follows easily.

\begin{lemma}
\label{lem:selection1}
Every ordered set $\{w_1,w_2\ldots w_r\}$ that satisfies property $\Q$ can be mapped to a \textit{selection} in such a way that for any two distinct ordered sets satisfying $\Q$, the corresponding \textit{selections} are distinct. 
\end{lemma}

\begin{proof}
We assign a \textit{selection} to an ordered set $\{w_1,w_2\ldots w_r\}$ satisfying $\Q$ using the algorithm below.

\bigskip
\fbox{
\parbox{\textwidth}{
\small

\bigskip 
\textbf{Initialization}
\begin{itemize}
\item Set $i=1$ and $b_i=0, x_i=w_i$.
\item While $(i\leq r)$, do: 
\item $i\rightarrow i+1$. Set $x_i=w_i$. Set $b_i=0$ if there is no $j<i$ such that $w_i\in \N(w_j)$.  Else set $b_i=1$. 
\item End while.
\end{itemize}

\textbf{Pointer creation}
\begin{itemize}
\item Define a relation $R:\{1,2\ldots r\}\rightarrow \{0,1,2\ldots r\}$ as follows. 
\item Set $i=1$. While $(i\leq r)$, do:
\item If $b_i=1$, set $R(i)=0$. If $b_i=0$, find the smallest $j>i$ such that $b_j=1$ and $x_i\in \N(x_j)$ (such a $j$ exists due to property $\Q$). Set $R(i)=j$. Set $i \rightarrow i+1$.
\item End while. 
\end{itemize}

\textbf{Update}
\begin{itemize}
\item Let $\mathcal{S}$ be the set of all subsets of $\{1,2,\ldots r\}$ which have cardinality at least $2$. 
\item For each element $S\in \mathcal{S}$, do: 
\item Let $s$ be the cardinality of $S$ and $i_1,i_2\ldots i_s$ be its elements arranged in increasing order. If it holds that $b_{i_1}=b_{i_1}=\ldots b_{i_s}=0$ and $R(i_1)=R(i_2)=\ldots R(i_s) > 0$: set $b_{i_2}=b_{i_3}=\ldots b_{i_s}=1$. 
\item End For.
\end{itemize}
}
}
\bigskip

We show that above algorithm terminates and assigns a \textit{selection} to each ordered set satisfying property $\Q$. 
\begin{enumerate}
\item Consider the running of algorithm during the step \textbf{Initialization}. Condition $1$ of a \textit{selection} holds: for every $i$ for which there is a $j<i$ such that $x_i\in \N(x_j)$, we have set $b_i=1$. But we haven't constructed a \textit{selection} yet, since condition $2$ may not be satisfied. 

\item After the step \textbf{Pointer creation}, it may be possible that there exist indices $i_1,i_2\ldots i_s$ (for some $s\leq r$) such that $b_{i_1}=b_{i_1}=\ldots b_{i_s}=0$, $R(i_1)=R(i_2)=\ldots R(i_s) > 0$ and $i_s>i_{s-1}>\ldots i_1$. In this case, we find that $w_{i_2}\in \N(w_{R_{i_2}}), w_{i_1}\in \N(w_{R_{i_1}})$. But $\N(w_{R_{i_1}})=\N(w_{R_{i_2}})$, which implies that $|\N(i_1)\cap \N(i_2)| > 0$. Similarly, $|\N(i_1)\cap \N(i_3)| > 0$, \ldots $|\N(i_1)\cap \N(i_s)| > 0$. 

Thus, the step \textbf{Update} sets $b_{i_2}=b_{i_3}=\ldots b_{i_s}=1$, recognizing the fact that each of the points $w_{i_2},w_{i_3}\ldots w_{i_s}$ satisfy the property that the neighbourhood of each of them intersects with $\N(w_{i_1})$. This ensures that condition $1$ of \textit{selection} is still satisfied.
 
\item After the step \textbf{Update} terminates, condition $2$ of \textit{selection} is now satisfied as well. We now have that for every $i$ with $b_i=0$, there is no other $i'$ such that $R(i)=R(i')$ and $b_i=b_{i'}=0$. Furthermore, $b_{R(i)}=1$. Thus, number of $i$ with $b_i=0$ is at most as large as the number of $j$ with $b_j=1$.
\end{enumerate}

Lemma follows as two distinct ordered tuples satisfying $\Q$ are not assigned the same \textit{selection}.
\end{proof}
 
Now we prove Lemma \ref{lem:combinatorial1}.

\begin{proof}[Proof of Lemma \ref{lem:combinatorial1}]
For $n\leq m^2r$, we clearly have $N_{k,m}(n,r)\leq n^r \leq (m^2nr)^{\frac{r}{2}} < (4m^2nr)^{\frac{r}{2}}$.

So we assume $n> m^2r$. We bound the number of \textit{selections}, which gives the desired upper bound on $N_{k,m}(n,r)$ using Lemma \ref{lem:selection}. 

Consider the collection of those \textit{selections}, for which number of $i$ such that $b_i=0$ is $u$. For each $i$ with $b_i=0$, number of possible choices of $x_i$ is $n$. For each $i$ with $b_i=1$,  number of possible choices of $x_i$ is at most $m^2r$ (as there are at most $m^2$ number of $x_j\in W_{k,m}$ that satisfy $|\N(x_i)\cap \N(x_j)|>0$ for a given $x_i$). Hence total number of such \textit{selections} is at most ${r \choose u}n^{u}(m^2 r)^{r-u}$. Since $u\leq \frac{r}{2}$, total number of \textit{selections} is at most $$\sum_{u=0}^{\frac{r}{2}}{r \choose u}n^{u}(m^2r)^{r-u} \leq \sum_{u=0}^{\frac{r}{2}}{r \choose u}n^{\frac{r}{2}}(m^2r)^{\frac{r}{2}}<2^r(m^2nr)^{\frac{r}{2}}.$$  

This proves the lemma.
\end{proof}

\end{document}